\documentclass{article}
\usepackage[utf8]{inputenc}

\usepackage{amssymb,latexsym,amsmath,epsfig,amsthm,amsmath,amscd, graphicx, verbatim, setspace,authblk} 

\makeatletter

\renewcommand\section{\@startsection {section}{1}{\z@}
{-30pt \@plus -1ex \@minus -.2ex}
{2.3ex \@plus.2ex}
{\normalfont\normalsize\bfseries\boldmath}}

\renewcommand\subsection{\@startsection{subsection}{2}{\z@}
{-3.25ex\@plus -1ex \@minus -.2ex}
{1.5ex \@plus .2ex}
{\normalfont\normalsize\bfseries\boldmath}}

\renewcommand{\@seccntformat}[1]{\csname the#1\endcsname. }

\makeatother

\theoremstyle{definition}

\theoremstyle{plain}
\newtheorem{thm}{Theorem}[section]
\newtheorem{lem}[thm]{Lemma}
\newtheorem{cor}[thm]{Corollary}
\newtheorem{conj}[thm]{Conjecture}

\newtheorem{prop}[thm]{Proposition}
\usepackage{mathtools}
\usepackage{amsthm}
\usepackage{amssymb}
\usepackage{amsfonts}
\usepackage{graphicx}
\usepackage{ytableau}
\usepackage{hyperref}
\usepackage[utf8]{inputenc}
\usepackage{todonotes}
\usepackage{tikz}
\usepackage{comment}

\newcommand{\outdeg}{\operatorname{outdeg}}
\newcommand{\indeg}{\operatorname{indeg}}

\title{Random processes for generating task-dependency graphs}
\author{Jesse Geneson and Shen-Fu Tsai}
\date{}

\begin{document}

\maketitle

\begin{abstract}
We investigate random processes for generating task-dependency graphs of order $n$ with $m$ edges and a specified number of initial vertices and terminal vertices. In order to do so, we consider two random processes for generating task-dependency graphs that can be combined to accomplish this task. In the $(x, y)$ edge-removal process, we start with a maximally connected task-dependency graph and remove edges uniformly at random as long as they do not cause the number of initial vertices to exceed $x$ or the number of terminal vertices to exceed $y$. In the $(x, y)$ edge-addition process, we start with an empty task-dependency graph and add edges uniformly at random as long as they do not cause the number of initial vertices to be less than $x$ or the number of terminal vertices to be less than $y$. In the $(x, y)$ edge-addition process, we halt if there are exactly $x$ initial vertices and $y$ terminal vertices. For both processes, we determine the values of $x$ and $y$ for which the resulting task-dependency graph is guaranteed to have exactly $x$ initial vertices and $y$ terminal vertices, and we also find the extremal values for the number of edges in the resulting task-dependency graphs as a function of $x$, $y$, and the number of vertices. Furthermore, we asymptotically bound the expected number of edges in the resulting task-dependency graphs. Finally, we define a random process using only edge-addition and edge-removal, and we show that with high probability this random process generates an $(x, y)$ task-dependency graph of order $n$ with $m$ edges.
\end{abstract}

\section{Introduction}

A \textit{task-dependency graph} is a directed acyclic graph where the vertices represent tasks and the edge $(a, b)$ represents that the completion of task $b$ depends on the completion of task $a$. Task-dependency graphs are a useful tool in the design of parallel programs. Each task represents a part of the program specified by the programmer, and the task-dependency graph $G$ provides a representation of the interdependencies between the parts of the program. Understanding the structure of $G$ can be useful for scheduling parts of the corresponding program on multiple processors and bounding the total running time. Several paradigms have been developed for scheduling the tasks in a task-dependency graph given multiple processors, including \textit{work-sharing} \cite{blumofe1, eager, mirchandaney, vanhoudt}, \textit{work-stealing} \cite{burton, halstead, blumofe, blumofe0, feldmann, finkel, halbherr, karp, kuszmaul, mohr, rudolph, vandevoorde}, and \textit{parallel depth-first scheduling} \cite{blelloch, blelloch0, chen}. 

In this paper, we focus on task-dependency graphs with a specified number of initial vertices and terminal vertices. An \textit{$(x, y)$ task-dependency graph} is a task-dependency graph with $x$ initial vertices and $y$ terminal vertices, where initial and terminal vertices are those with zero in-degree and out-degree, respectively. A \textit{minimal $(x, y)$ task-dependency graph} is an $(x, y)$ task-dependency graph for which removing any edge produces a task-dependency graph that is not an $(x, y)$ task-dependency graph. The family of $(1, 1)$ task-dependency graphs have been used to analyze project management methods for military programs \cite{randgbsd}. Specifically, Sentinel (formerly called Ground Based Strategic Deterrent) is a program of the United States Air Force which designs and produces more than 600 missiles, refurbishes 450 silos, and designs and produces at least 24 launch centers, among many other tasks \cite{sentinel}. In a recent report \cite{randgbsd}, task-dependency graphs were used to represent the interdependent tasks in Sentinel. In order to analyze project management methods for Sentinel including the Program Evaluation Review Technique (PERT) \cite{pert} and the Critical Path Method (CPM) \cite{cpm}, the report \cite{randgbsd} modeled Sentinel with $(1, 1)$ task-dependency graphs and applied the project management methods to these task-dependency graphs. They noted that for any task-dependency graph $G$, it is possible to turn $G$ into a $(1, 1)$ task-dependency graph. Simply add a new initial vertex $u$ and terminal vertex $v$, make $u$ the only initial vertex by adding edges from $u$ to all of the initial vertices of $G$, and make $v$ the only terminal vertex by adding edges from all of the terminal vertices of $G$ to $v$.

As part of the analysis in \cite{randgbsd}, random $(1, 1)$ task-dependency graphs with $n$ vertices and $m$ edges were generated using the following process. Start with a maximally connected task-dependency graph, and then remove edges uniformly at random until $m$ edges remain. Whenever an edge removal causes the task-dependency graph to have more than one initial vertex or more than one terminal vertex, cancel the edge removal. Note that if any round of this process produces a minimal $(1, 1)$ task-dependency graph, then no more edges can be removed. Clearly the minimum possible number of edges in a $(1, 1)$ task-dependency graph of order $n$ is $n-1$, but in \cite{randgbsd} it was noted that the random process can generate minimal $(1, 1)$ task-dependency graphs of order $n$ that have more than $n-1$ edges. So, if you pick a value of $m$ close to $n-1$, it is possible that the random process in \cite{randgbsd} does not terminate. Based on this observation, a natural problem is to determine how many edges can be in a minimal $(1, 1)$ task-dependency graph of order $n$. More generally, how many edges can be in a minimal $(x,y)$ task-dependency graph of order $n$? Note that these questions do not require the underlying graph to be connected, so we also investigate the corresponding questions for connected graphs.

In order to answer these questions, we define another random process for generating task-dependency graphs called the \textit{$(x, y)$ edge-removal process} which is almost the same as the random process described in the last paragraph, except it has a different condition for termination. Start with a maximally connected task-dependency graph, and then remove edges uniformly at random. Whenever an edge removal causes the task-dependency graph to have more than $x$ initial vertices or more than $y$ terminal vertices, cancel the edge removal. The process terminates when it is impossible to remove any more edges. Note that this differs from the process in \cite{randgbsd} in that the $(x, y)$ edge-removal process terminates when no more edges can be removed, while the process in \cite{randgbsd} terminated when the number of edges was equal to a specified value. We prove for $n > x+y$ that the minimum possible number of edges in any $(x,y)$ task-dependency graph produced by the $(x,y)$ edge-removal process is $n-\min(x,y)$. Then we prove for $n > x+y$ that the maximum possible number of edges in a task-dependency graph produced by the $(x, y)$ edge-removal process on $n$ vertices is $2n-x-y-2$, which we show is equal to the maximum possible number of edges in a minimal $(x,y)$ task-dependency graph of order $n$ for all $x, y \ge 1$. In particular, we show that the maximum possible number of edges of a minimal $(1,1)$ task-dependency graph of order $n$ is $2n-4$. Thus for values of $m$ less than $2n-4$, it is possible for the random process in \cite{randgbsd} not to yield a $(1, 1)$ task-dependency graphs with $n$ vertices and $m$ edges. On the other hand, for $m \ge 2n-4$, the random process in \cite{randgbsd} is guaranteed to yield a $(1, 1)$ task-dependency graphs with $n$ vertices and $m$ edges. We also characterize the minimal $(x,y)$ task-dependency graphs of order $n$ which attain the maximum number of edges, and we show in particular that any directed paths in these graphs have length at most $2$. 

We also consider a random process for generating $(x,y)$ task-dependency graphs which adds edges instead of removing them. Suppose that we begin with an empty directed graph with vertex set $\left\{1,2,\dots, n\right\}$, and then we start uniformly at random adding edges $(a, b)$ with $a < b$. If any edge causes there to be fewer than $x$ initial vertices or fewer than $y$ terminal vertices, then we cancel the addition of that edge. We halt the process if the result after any edge addition is an $(x,y)$ task-dependency graph, or if we cannot add any more edges. We call this the \textit{$(x, y)$ edge-addition process}. We prove for all $x, y \ge 1$ and $n > x+y$ that the minimum possible number of edges in any $(x,y)$ task-dependency graph produced by the $(x,y)$ edge-addition process on $n$ vertices is $n-\min(x, y)$ and the maximum possible number of edges in any task-dependency graph produced by the $(x,y)$ edge-addition process is $\binom{n}{2}+1-\binom{\max(x,y)}{2}-\binom{\min(x,y)+1}{2}$. We also prove that the expected number of edges in the $(x, y)$ edge-addition process on $n$ vertices is $\Theta(n^2)$, where the constant in the lower bound depends on $(x,y)$. For both the $(x,y)$ edge-addition process and $(x,y)$ edge-removal process, we also characterize the pairs $(x,y)$ for which each process is guaranteed to terminate with an $(x,y)$ task-dependency graph for $n > \max(x,y)$. For every pair $(x, y)$, we show that the probability that the $(x,y)$ edge-addition process on $n$ vertices terminates with an $(x,y)$ task-dependency graph goes to $1$ as $n$ goes to infinity.

In Section~\ref{ex_bounds}, we determine for $x,y\geq1$ that the maximum possible number of edges of a minimal $(x,y)$ task-dependency graph of order $n$ is $0$ if $n = \max(x, y)$ (in which case $x = y$), $2n-x-y-1$ if $n=\max(x, y)+1$, and $2n-x-y-2$ if $n>\max(x, y)+1$. Furthermore, we characterize the minimal $(x,y)$ task-dependency graphs of order $n$ with the maximum possible number of edges. As a corollary, we show that every minimal $(x,y)$ path-dependency graph of order $n$ with the maximum number of edges has no directed path with length more than $2$. In Section~\ref{eap}, we focus on the edge-addition process on $n$ vertices. We determine the pairs $(x,y)$ for which the $(x, y)$ edge-addition process terminates with an $(x,y)$ task-dependency graph, we prove extremal bounds for the number of edges in the task-dependency graphs generated by the $(x, y)$ edge-addition process, we show that the expected number of isolated vertices is $o(1)$, and we also bound the expected number of edges. 

In Section~\ref{s:experiment}, we present experimental results obtained from implementing the $(x, y)$ edge-removal and $(x, y)$ edge-addition processes in Python. We make several conjectures on the expected number of edges, expected maximum directed path length, and probability of generating an $(x, y)$ task-dependency graph for both processes. In particular, our experimental results suggest that the probability of generating an $(x, y)$ task-dependency graph approaches $1$ for both the $(x, y)$ edge-removal process and $(x, y)$ edge-addition process as the number of vertices increases, and we prove that this is true for the $(x, y)$ edge-addition process. In Section~\ref{open_prob}, we define a random process using only edge-addition and edge-removal which with high probability generates an $(x, y)$ task-dependency graph of order $n$ with $m$ edges. This process combines the $(x, y)$ edge-addition process with the $(x, y)$ edge-removal process. Note that this process provides a method different from the one in \cite{randgbsd} for randomly generating $(1, 1)$ task-dependency graphs of order $n$ with $m$ edges. In the same section, we determine the maximum possible number of orderings of the tasks in a $(1,1)$ task-dependency graph of order $n$. Interestingly, this maximum is attained by minimal $(1,1)$ task-dependency graphs of order $n$ with the maximum possible number of edges. We also determine the maximum possible number of orderings of the tasks in an $(x,y)$ task-dependency graph of order $n$ when $x,y \ge 1$ and $\max(x, y) \le n \le \max(x,y)+2$. Furthermore, we discuss some open problems and future directions.

\section{Edge-removal process}\label{ex_bounds}

In this section, we focus on the $(x, y)$ edge-removal process. We show that the process does not always result in an $(x, y)$ task-dependency graph. We also examine the extremal behavior of the $(x, y)$ edge-removal process. Specifically, we determine the maximum possible number of edges in a minimal $(x, y)$ task-dependency graph of order $n$ for all $n \ge \max(x, y)$, which is equivalent to the maximum possible number of edges in an $(x, y)$ task-dependency graph obtained from the $(x, y)$ edge-removal process on $n$ vertices. We also characterize the minimal $(x, y)$ task-dependency graphs of order $n$ with the maximum possible number of edges. 

In order to obtain the results in this section, we introduce some terminology. Given a task-dependency graph $G$, we call a vertex $v \in V(G)$ an \textit{initial vertex} if there is no $u \in V(G)$ with $(u, v) \in E(G)$. We call $v \in V(G)$ a \textit{terminal vertex} if there is no $u \in V(G)$ with $(v, u) \in E(G)$. Initial vertices are the same as roots and sources, and terminal vertices are the same as tips and sinks. Note that it is possible for a vertex $v \in V(G)$ to be both initial and terminal, in which case $v$ is isolated. A vertex $v \in V(G)$ is \textit{exterior} if it is an initial vertex or a terminal vertex, and it is \textit{interior} otherwise.

It is clear that the $(1, 1)$ edge-removal process always results in a $(1,1)$ task-dependency graph, since a maximally-connected directed acyclic graph is a $(1,1)$ task-dependency graph. In the next proposition, we prove that this generalizes to the $(x, x)$ edge-removal process for all $x \ge 1$. However, in the proposition after that, we show that the $(x, y)$ edge-removal process does not necessarily result in an $(x, y)$ task-dependency graph when $x \neq y$.

\begin{prop}\label{xxedgeremove}
    For all $x \geq 1$ and $n \ge x$, the $(x, x)$ edge-removal process always results in an $(x, x)$ task-dependency graph.
\end{prop}

\begin{proof}
Suppose that we completely run the $(x, x)$ edge-removal process and obtain the task-dependency graph $G$ with $r$ initial vertices and $s$ terminal vertices. To show that $r = s = x$, assume for contradiction that $r < x$ or $s < x$.  Without loss of generality, let $r < x$. Note that we must have $s = x$, or else we could remove another edge and the resulting task-dependency graph would still have at most $x$ initial vertices and at most $x$ terminal vertices. We first argue that $G$ must have at least $2$ non-isolated terminal vertices in the same component. Indeed, if every component of $G$ had at most one terminal vertex, then we would have $r \geq s$ since every component must also have at least one initial vertex, contradicting the assumption that $r < x$ and $s = x$. Thus some component $C$ of $G$ has multiple non-isolated terminal vertices. 

Since $C$ has multiple non-isolated terminal vertices, there must exist some vertex $u$ in $C$ with at least two outgoing edges. Let $u$ have edges to vertices $v$ and $w$. If we remove the edge $(u, v)$ from $G$, the resulting task-dependency graph has at most $r+1$ initial vertices and $s$ terminal vertices. This contradicts our assumption that the $(x, x)$ edge-removal process terminated on the task-dependency graph $G$, since it was still possible to remove an edge while keeping the number of initial vertices at most $x$ and the number of terminal vertices at most $x$. Thus the $(x, x)$ edge-removal process must result in an $(x, x)$ task-dependency graph.
\end{proof}

In the next proposition, we consider the $(x, y)$ edge-removal process for $x \neq y$, and we show that it is possible for the process to halt on task-deperndency graphs which do not have $x$ initial vertices and $y$ terminal vertices.

\begin{prop}
For all $x, y\ge 1$ with $x \neq y$ and $n \ge \max(x, y)$, the $(x, y)$ edge-removal process does not necessarily result in an $(x, y)$ task-dependency graph.
\end{prop}

\begin{proof}
Without loss of generality, suppose that $x > y$, so $x \geq 2$. Let $G_{y, n}$ be the task-dependency graph on $\left\{1,2,\dots, n\right\}$ where $\left\{1,\dots, n-y+1\right\}$ form a directed path and there are $y-1$ isolated vertices $n-y+2, \dots, n$ (and there are no isolated vertices if $y = 1$). Then $G_{y,n}$ has $y$ initial vertices and $y$ terminal vertices. If we remove any edge, it would result in a task-dependency graph with $y+1$ terminal vertices. Thus it is impossible for the $(x, y)$ edge-removal process to result in an $(x, y)$ task-dependency graph if any round produces $G_{y,n}$. 
\end{proof}

Next we focus on finding the extremal values for the number of edges in $(x, y)$ task-dependency graphs obtained from the $(x, y)$ edge-removal process on $n$ vertices. We start with the minimum, which is clearly $n-1$ when $x = y = 1$. There cannot be fewer than $n-1$ edges, since otherwise there would be multiple components, which would imply multiple initial vertices and multiple terminal vertices. Furthermore the directed path of order $n$ has $n-1$ edges. We generalize this to all $x, y \ge 1$ by showing that the minimum possible number of edges in any $(x,y)$ task-dependency graph produced by the $(x,y)$ edge-removal process is $n-\min(x, y)$ for all $n > \max(x, y)$. We start with a lemma.

\begin{lem}\label{minedgelemma}
  For all $x, y \ge 1$ and $n > \max(x,y)$, the minimum possible number of edges in any $(x,y)$ task-dependency graph is $n-\min(x,y)$.  
\end{lem}

\begin{proof}
Without loss of generality, let $x \ge y$. For the upper bound, let $G$ be the task-dependency graph on $1, \dots, n$ with isolated vertices $n-y+2, \dots, n$ (and no isolated vertices if $y = 1$), a directed path on $x-y+2, \dots, n-y+1$, and edges $(a, x-y+2)$ for each $a$ with $1 \le a \le x-y+1$. Since the $y-1$ isolated vertices count as both initial and terminal vertices, $G$ has $x$ initial vertices, $y$ terminal vertices, and $n-y$ edges. Thus the minimum possible number of edges in any $(x,y)$ task-dependency graph is at most $n-y$.

For the lower bound, suppose that $H$ is a task-dependency graph of order $n$ with at most $n-y-1$ edges. Then $H$ has at least $y+1$ components, so $H$ has at least $y+1$ terminal vertices. Thus $H$ cannot be an $(x,y)$ task-dependency graph, so the number of edges in any $(x,y)$ task-dependency graph is always greater than $n-y-1$.
\end{proof}

\begin{cor}\label{eamin}
    For all $x, y \ge 1$ and $n > \max(x,y)$, the minimum possible number of edges in any $(x,y)$ task-dependency graph produced by the $(x,y)$ edge-removal process is $n-\min(x,y)$.  
\end{cor}

\begin{proof}
    The lower bound follows from Lemma~\ref{minedgelemma}. For the upper bound, note that we can use the same task-dependency graph $G$ from the proof of Lemma~\ref{minedgelemma}.
\end{proof}

As a second corollary, we show the same result for any task-dependency graph produced by the $(x,y)$ edge-removal process. Note that we have removed the requirement from Corollary~\ref{eamin} that the task-dependency graph is an $(x, y)$ task-dependency graph.

\begin{cor}\label{eamintask}
  For all $x, y \ge 1$ and $n > \max(x,y)$, the minimum possible number of edges in any task-dependency graph produced by the $(x,y)$ edge-removal process is $n-\min(x,y)$.  
\end{cor}

\begin{proof}
Suppose that $G$ is a task-dependency graph produced by the $(x, y)$ edge-removal process on $n$ vertices. Since the edge-removal process never decreases the number of initial vertices or terminal vertices in any round, $G$ must be an $(r, s)$ task-dependency graph for some $r \le x$ and $s \le y$. If $r = x$ and $s = y$, then $G$ has at least $n-\min(x,y)$ edges by Corollary~\ref{eamin}. Otherwise $r \le x-1$ or $s \le y-1$. In either case, $G$ has at least $n-\min(x,y)$ edges by Lemma~\ref{minedgelemma}. 
\end{proof}

For the remainder of this section, we focus on finding the maximum possible number of edges in $(x, y)$ task-dependency graphs obtained from the $(x, y)$ edge-removal process on $n$ vertices, which is an equivalent problem to determining the maximum possible number of edges in a minimal $(x,y)$ task-dependency graph of order $n$. As a corollary, we obtain asymptotic bounds on the expected number of edges in the random task-dependency graph obtained from the $(x, y)$ edge-removal process on $n$ vertices. In order to determine the maximum possible number of edges in a minimal $(x,y)$ task-dependency graph of order $n$, we prove several lemmas. 

\begin{lem}\label{edge_remove}
Graph $G$ is a minimal $(x, y)$ task-dependency graph if and only if for every $(u,v) \in E(G)$, $\outdeg(u) = 1$ or $\indeg(v) = 1$.
\end{lem}

\begin{proof}
Suppose that $G$ is a minimal $(x, y)$ task-dependency graph and  $\outdeg(u) > 1$ and $\indeg(v) > 1$ for some $(u,v)\in E(G)$. Then there exist vertices $w, z \in V(G)$ such that $(u, w) \in E(G)$ and $(z, v) \in E(G)$. Since $G$ is an $(x, y)$ task-dependency graph, there is a path from $w$ to at least one of the terminal vertices of  $G$, and there is a path from at least one of the initial vertices of $G$ to $z$. Thus if we remove the edge $(u, v)$ from $E(G)$, then there is still a path from at least one of the initial vertices to $v$ that passes through $z$, and there is still a path from $u$ to at least one of the terminal vertices that passes through $w$.

If for every $(u,v)\in E(G)$ we have $\outdeg(u) = 1$ or $\indeg(v) = 1$, then removing $(u,v)$ makes $u$ a new initial vertex or $v$ a new terminal vertex. Thus $G$ is a minimal $(x,y)$ task-dependency graph.
\end{proof}

We call a path $v_1,v_2,\ldots,v_k$ with $k \geq 3$ \textit{removable} if $\indeg(v_i)=\outdeg(v_i)=1$ for $1<i<k$ and $\outdeg(v_1)>1,\indeg(v_k)>1$. By \textit{removing} a removable path we mean to delete the edges $(v_1,v_2),\ldots,(v_{k-1},v_k)$ and the vertices $v_2,\ldots,v_{k-1}$. 

\begin{lem}\label{remove_p}
If $G$ is a minimal $(x, y)$ task-dependency graph with a removable path $P$, then the graph $G'$ obtained from $G$ by removing $P$ is also a minimal $(x,y)$ task-dependency graph.
\end{lem}

\begin{proof}
Clearly the sufficient condition for $G'$ to be a minimal $(x,y)$ task-dependency graph in Lemma~\ref{edge_remove} still holds, so $G'$ is a minimal task-dependency graph.
\end{proof}

\begin{lem}\label{p_to_rp}
Suppose that $G$ is a minimal $(x, y)$ task-dependency graph that has a path $(v_1,\ldots,v_k)$ with $k\geq 3$, $\outdeg(v_1)>1$, and $\indeg(v_k)>1$. Then $G$ has a removable path.
\end{lem}

\begin{proof}
We prove the lemma by induction on $k$. If $k = 3$, then it follows immediately from Lemma~\ref{edge_remove} that any path $(v_1,v_2,v_3)$ with $\outdeg(v_1) > 1$ and $\indeg(v_3) = 1$ must be a removable path.

Now suppose for some $k \geq 3$ that the lemma is true for all $3 \leq j \leq k$ and that $G$ has a path $(v_1,..,v_{k+1})$ with $\outdeg(v_1) > 1$ and $\indeg(v_{k+1})> 1$. By Lemma~\ref{edge_remove}, $\indeg(v_2) = 1$. If $\outdeg(v_2) > 1$, then $(v_2, \dots, v_{k+1})$ contains a removable path by inductive hypothesis.

Otherwise if $\outdeg(v_2) = 1$, then let $i > 2$ be minimal such that $\indeg(v_i) > 1$ or $\outdeg(v_i) = 1$. If $\indeg(v_i) = 1$, then $(v_1, \dots, v_i)$ is a removable path. If $\outdeg(v_i) > 1$, then $i < k$ by Lemma~\ref{edge_remove}. Thus $(v_i, \dots, v_{k+1})$ contains a removable path by inductive hypothesis.    
\end{proof}

\begin{cor}\label{underlying}
    If $G$ is a minimal $(x, y)$ task-dependency graph with no removable path, then the underlying undirected graph of $G$ is a disjoint union of trees.
\end{cor}

\begin{proof}
Suppose for contradiction that the underlying graph of $G$ has a cycle $C$. Partition the edges of this cycle $C$ into maximal directed paths, i.e., directed paths of maximum possible length. Since $G$ is acyclic, the partition includes at least two maximal directed paths. If any of the maximal directed paths has only one edge $e=(u,v)$, then $\outdeg(u)>1$ and $\indeg(v)>1$ by maximality, contradicting Lemma~\ref{edge_remove}. So $G$ has a path $(v_1,\ldots,v_k)$ with $k\geq 3$, $\outdeg(v_1)>1$, and $\indeg(v_k)>1$ by maximality. By Lemma~\ref{p_to_rp}, $G$ must have a removable path, contradicting the assumption.
\end{proof}

We define two families of minimal $(x,y)$ task-dependency graphs which we will use in the following proofs. For each $x \geq 1$, $y \geq 1$, and $n \geq x+2$, define $S_{x,y,n}$ to be the minimal $(x,y)$ task-dependency graph of order $n$ with $y-1$ isolated vertices, $x-y+1$ non-isolated initial vertices $u_1, \dots, u_{x-y+1}$, $1$ non-isolated terminal vertex $v$, $n-x-1$ interior vertices $i_1, \dots, i_{n-x-1}$, and edges $(u_1, i_t)$ for $1 \leq t \leq n-x-1$, $(i_t, v)$ for $1 \leq t \leq n-x-1$, and $(u_t, v)$ for $2 \leq t \leq x-y+1$. For each $x \geq 1$, $y \geq 1$, and $n \geq x+y+1$, define $T_{x,y,n}$ to be the minimal $(x,y)$ task-dependency graph of order $n$ with $x$ initial vertices $u_1, \dots, u_x$, $y$ terminal vertices $v_1, \dots, v_y$, $n-x-y$ interior vertices $i_1, \dots, i_{n-x-y}$, and edges $(u_1, i_t)$ for $1 \leq t \leq n-x-y$, $(i_t, v_1)$ for $1 \leq t \leq n-x-y$, $(u_1, v_t)$ for $2 \leq t \leq y$, and $(u_t, v_1)$ for $2 \leq t \leq x$.

\begin{thm}\label{mainthm}
For $x,y\geq1$, the maximum possible number of edges of a minimal $(x,y)$ task-dependency graph of order $n$ is $0$ if $n = \max(x, y)$ (in which case $x=y$), $2n-x-y-1$ if $n=\max(x, y)+1$, and $2n-x-y-2$ if $n>\max(x, y)+1$.
\end{thm}
\begin{proof}
Without loss of generality, suppose that $x \ge y$. If $n = x$, then all vertices in any $(x,y)$ task-dependency graph of order $n$ are both initial vertices and terminal vertices, so $x = y$ and there are no edges. If $y < x$,  then any $(x,y)$ task-dependency graph $G$ of order $n$ has at most $y-1$ isolated vertices, each of which counts as both an initial vertex and a terminal vertex, i.e., $G$ has at least $x+1$ exterior vertices. Thus if the order of $G$ is $x+1$, then $G$ has $y-1$ isolated vertices, one non-isolated terminal vertex, $x-y+1$ non-isolated initial vertices, no interior vertex, and $x-y+1=2n-x-y-1$ edges, with one edge from each non-isolated initial vertex to the single non-isolated terminal vertex.

For $n>x+1$ we prove the statement by induction on $n$. If $G$ is an $(x,y)$ task-dependency graph of order $n = x+2$, then $G$ either has (a) $y-1$ isolated vertices, $x-y+1$ non-isolated initial vertices, one non-isolated terminal vertex, and one interior vertex, or (b) $y-2$ isolated vertices, $x-y+2$ non-isolated initial vertices, two non-isolated terminal vertices, and no interior vertex. Note that case (b) can only happen when $y > 1$. In either case $G$ has $2n-x-y-2$ edges. 

Now suppose that $n>x+2$, and let $G$ be a minimal $(x,y)$ task-dependency graph of order $n$. If $G$ has a removable path $P=(v_1,\ldots,v_k)$, then
we remove it to obtain a new $(x,y)$ task-dependency graph $G'$ by Lemma~\ref{remove_p}. $G'$ cannot have exactly $x+1$ vertices, as otherwise every non-isolated vertex in $G'$ is exterior and $(v_1,v_k)\in E(G)$, i.e., $G$ is still an $(x,y)$ task-dependency graph after removing $(v_1,v_k)$, a contradiction of its minimality. Therefore $G'$ has at least $x+2$ vertices and thus by inductive assumption at most $2(n-k+2)-x-y-2$ edges. Thus, graph $G$ has at most $2(n-k+2)-x-y-2+k-1\leq 2n-x-y-2$ edges since $k \ge 3$.

If $G$ does not have a removable path, then the underlying undirected graph of $G$ must be a disjoint union of trees by Corollary~\ref{underlying}. If $G$ has $z$ isolated vertices, then it must have at least $z+1$ components since $n \geq x+3$.

Suppose that $G$ has $z+1$ components, and let $C$ be the component of $G$ that is not an isolated vertex. Then the underlying undirected graph of $C$ must be a tree. Thus this component must have at least one interior vertex since $n \geq x+3$ and $x \geq y$. Furthermore, $G$ has $n-(z+1)$ edges. Clearly $x+y-z\leq n-1$ in this case, so $G$ has at most $2n-x-y-2$ edges.

Otherwise $G$ has at least $z+2$ components. In this case, $G$ has at most $n-(z+2)$ edges. Since $x+y-z\leq n$, we conclude that $G$ has at most $2n-x-y-2$ edges.

Note that the upper bound of $2n-x-y-2$ edges is attained by the family of graphs $S_{x,y,n}$ for all $n \geq x+2$. 
\end{proof}

Using the last theorem, we can also determine the maximum possible number of edges in any task-dependency graph produced by the $(x, y)$ edge-removal process on $n$ vertices. Note that we do not require the produced task-dependency graph to have $x$ initial vertices and $y$ terminal vertices in the next corollary.

\begin{cor}
For $x, y \ge 1$ and $n > x+y$, the maximum possible number of edges in a task-dependency graph produced by the $(x, y)$ edge-removal process on $n$ vertices is $2n-x-y-2$.
\end{cor}

\begin{proof}
    Suppose that $G$ is a task-dependency graph produced by the $(x, y)$ edge-removal process on $n$ vertices, so $G$ must be an $(r, s)$ task-dependency graph for some $r \le x$ and $s \le y$, If $r = x$ and $s = y$, then $G$ has at most $2n-x-y-2$ edges. Otherwise, we have (a) $r \le x-1$ and $s = y$ or (b) $r = x$ and $s \le y-1$. In case (a), there cannot exist vertices $u, v, w$ with edges $(u, v)$ and $(u, w)$, or else we could remove the edge $(u, v)$ to increase the number of initial vertices by at most one without increasing the number of terminal vertices. Then the underlying undirected graph of $G$ is a disjoint union of trees, so $G$ has at most $n-1 \le 2n-x-y-2$ edges since $n > x+y$. The proof for case (b) is analogous.
\end{proof}

Note that any task-dependency graph resulting from the $(x, y)$ edge-removal process must be a minimal $(r, s)$ task-dependency graph for some $r \leq x$ and $s \leq y$. If it was not minimal, then we would be able to remove another edge and still obtain a task-dependency graph with at most $x$ initial vertices and at most $y$ terminal vertices. Thus the $(x, y)$ edge-removal process on $n$ vertices must result in a task-dependency graph with at most $2n$ edges by Theorem~\ref{mainthm} and at least $n - \min(x, y)$ edges by Corollary~\ref{eamintask}. This implies the following corollary about the expected number of edges.

\begin{cor}
    For all $x, y \ge 1$, the expected number of edges in the random task-dependency graph resulting from the $(x, y)$ edge-removal process on $n$ vertices is $\Theta(n)$.
\end{cor}

In the next theorem, we characterize the minimal $(x,y)$ task-dependency graphs with the maximum possible number of edges for all $n \ge \max(x, y)$.

\begin{thm}\label{structure-thm}
A graph $G$ of order $n$ has the maximum possible number of edges among all minimal $(x,y)$ task-dependency graphs of order $n$ if and only if it has at most two components that are not isolated vertices, and 
\begin{enumerate}
    \item if $G$ consists of all isolated vertices, then $n = x = y$, \item if $G$ has two components that are not isolated vertices, then $G$ has no interior vertex, and
    \item if $G$ has exactly one component $C$ that is not an isolated vertex, then either
    \begin{enumerate}
        \item $C$ has no interior vertex, 
        \item there exist an initial vertex $u$, a terminal vertex $v$, and multiple interior vertices in $C$ such that all interior vertices of $C$ are adjacent to both $u$ and $v$, $u$ may be adjacent to more terminal vertices, and $v$ may be adjacent to more initial vertices, or 
        \item $C$ has exactly one interior vertex, which is adjacent to $p>0$ initial vertices and $q>0$ terminal vertices. If $p=1$, then this initial vertex may be adjacent to more terminal vertices. If $q=1$, then this terminal vertex may be adjacent to more initial vertices. If $p > 1$, then these initial vertices are only adjacent to the interior vertex. If $q > 1$, then these terminal vertices are only adjacent to the interior vertex.
    \end{enumerate}
\end{enumerate}
\end{thm}
\begin{proof}
We prove the forward direction by induction, following the induction in the proof of Theorem~\ref{mainthm}. Suppose that $G$ is a graph of order $n$ with the maximum possible number of edges among all minimal $(x,y)$ task-dependency graphs of order $n$. Without loss of generality we assume that $x\geq y$. If $n = x$, then $x = y$ and the graph consists of all isolated vertices, i.e., the structure in (1). If $n = x+1$, then $G$ has exactly one component $C$ that is not an isolated vertex, and $C$ has no interior vertex, i.e., the structure in (3)(a). If $n=x+2$, then $G$ either has exactly two  non-isolated terminal vertices and no interior vertices, i.e., the structure in (2), or exactly one non-isolated terminal vertex and exactly one interior vertex. These two vertices must belong to the same component, so $G$ has the structure in (3)(c) with $q = 1$.

Now suppose that $n > x+2$. If $G$ does not have a removable path, then from the last two paragraphs of the proof of Theorem~\ref{mainthm} $G$ should either have no interior vertex and exactly two components that are not isolated vertices, i.e., the structure in (2), or exactly one component $C$ that is not an isolated vertex with exactly one interior vertex. If $C$ has multiple initial vertices adjacent to its interior vertex, then all of these initial vertices are only adjacent to the interior vertex, or else $G$ would not be a minimal $(x,y)$ task-dependency graph. Similarly if $C$ has multiple terminal vertices adjacent to its interior vertex, then all of these terminal vertices are only adjacent to the interior vertex, or else $G$ would not be a minimal $(x,y)$ task-dependency graph. Thus $G$ has the structure in (3)(c). 

If $G$ has a removable path $P=(v_1,\ldots,v_k)$ whose removal produces $G'$, then from the proof of Theorem~\ref{mainthm} we know that $k=3$. Moreover, $v_1$ and $v_k$ are not adjacent since $G$ is a minimal $(x,y)$ task-dependency graph, so either in $G'$ they belong to different components or in the underlying undirected graph of $G$ they belong to a cycle. If they belong to a cycle in the underlying undirected graph of $G$, then we partition the edges of this cycle into maximal directed paths. We claim that this partition consists of exactly two $3$-vertex directed paths, and thus by inductive hypothesis $G'$ has the structure in (3)(b) or (3)(c). First note that the partition consists of at least two maximal directed paths, since $G$ is acyclic. If any of these directed paths has only two vertices, then the edge connecting them violates Lemma~\ref{edge_remove}, so each of the directed paths has at least three vertices and thus at least one interior vertex. Since it is impossible to have exactly three maximal directed paths in a cycle, suppose for contradiction that the partition consists of four or more directed paths. Then $G'$ has three interior vertices in a structure not present in (3)(b) and (3)(c), contradicting the inductive hypothesis. Hence the partition consists of two maximal directed paths. Neither directed path has more than three vertices, or else $G'$ would contain a directed path with more than three vertices, contradicting the inductive hypothesis. Thus both maximal directed paths have exactly three vertices. Therefore $G'$ has the structure in (3)(b) or (3)(c).

If $G’$ has the structure in (3)(b), then the initial vertex adjacent to the interior vertices in $G’$ must be $v_1$ and the terminal vertex adjacent to the interior vertices in $G’$ must be $v_k$, since the only vertices with distance $2$ in $G’$ are $v_1$ and $v_k$. Since $G$ is obtained from $G’$ by adding a single vertex $v_2$ that is adjacent to only $v_1$ and $v_k$, we conclude that $G$ has the structure in (3)(b). Now suppose that $G’$ has the structure in (3)(c). If multiple initial vertices were adjacent to the interior vertex in $G'$, then $G$ would not be a minimal $(x,y)$ task-dependency graph since we would be able to remove the edge from $v_1$ to the interior vertex in $G'$ and the resulting graph would still be an $(x,y)$ task-dependency graph. Thus only one initial vertex is adjacent to the interior vertex in $G'$. Similarly, only one terminal vertex is adjacent to the interior vertex in $G'$. Therefore $G$ has the structure in (3)(b). 

If in $G'$ the vertices $v_1$ and $v_k$ belong to different components, then both components are not isolated vertices because $\outdeg(v_1) > 1$ and $\indeg(v_k) > 1$ in $G$. Hence by inductive hypothesis, $G'$ has no interior vertices, $v_1$ is an initial vertex, and $v_k$ is a terminal vertex. Thus $G$ has exactly one component that is not an isolated vertex, and that component has the structure described in (3)(c) with $p = q = 1$.

For the backward direction, first suppose that $G$ consists of all isolated vertices, and $n = x = y$. Then clearly $G$ is a graph of order $n$ with the maximum possible number of edges among all minimal $(x,y)$ task-dependency graphs of order $n$. Next, suppose that $G$ has two components that are not isolated vertices, and $G$ has no interior vertex. If the components of $G$ that are not isolated vertices have size $j$ and $k$, then $G$ has $n-j-k$ isolated vertices and $j+k-2$ edges, and $x+y = 2n-j-k$. Thus $n \ge x+2$ and $G$ has $2n-x-y-2$ edges, so $G$ is a graph of order $n$ with the maximum possible number of edges among all minimal $(x,y)$ task-dependency graphs of order $n$. Next, suppose that $G$ has exactly one component $C$ that is not an isolated vertex and $C$ has no interior vertex, i.e., $G$ has the structure in (3)(a). If $C$ has size $j$, then $G$ has $n-j$ isolated vertices and $j-1$ edges, and $x+y = 2n-j$. Thus $n = x+1$ and $G$ has $2n-x-y-1$ edges, so $G$ is a graph of order $n$ with the maximum possible number of edges among all minimal $(x,y)$ task-dependency graphs of order $n$. 

Next, suppose that $G$ has the structure in (3)(b), i.e., it has $i\geq 0$ isolated vertices and exactly one component $C$ that is not an isolated vertex, and $C$ has $j>1$ interior vertices all adjacent to an initial vertex $u$ and terminal vertex $v$, $u$ is adjacent to $k\geq 0$ other terminal vertices, and $v$ is adjacent to $l\geq 0$ initial vertices.
Then $x=i+1+l$, $y=i+1+k$, and $n=i+2+k+l+j$. Clearly $n\geq\max(x,y)+3$ and $G$ has $2j+k+l=2n-x-y-2$ edges. Hence by Theorem~\ref{mainthm}, $G$ is a graph of order $n$ with the maximum possible number of edges among all minimal $(x,y)$ task-dependency graphs of order $n$. 

Finally, suppose that $G$ has the structure in (3)(c), so the component $C$ has exactly one interior vertex which is adjacent to $p > 0$ initial vertices and $q > 0$ terminal vertices. We consider three cases. First, suppose that $p > 1$, $q > 1$, and the initial vertices and terminal vertices in $C$ are only adjacent to the interior vertex. Then $G$ has $n-p-q-1$ isolated vertices, $x = n-q-1$, $y = n-p-1$, and $G$ has $p+q$ edges. Thus $G$ has $2n-x-y-2$ edges, so $G$ is a graph of order $n$ with the maximum possible number of edges among all minimal $(x,y)$ task-dependency graphs of order $n$. 

Second, suppose that $p = 1$, $q = 1$, the initial vertex adjacent to the interior vertex in $C$ is adjacent to $r \geq 0$ terminal vertices, and the terminal vertex adjacent to the interior vertex in $C$ is adjacent to $s \geq 0$ initial vertices. Then $G$ has $n-r-s-3$ isolated vertices, $x = n-r-2$, $y = n-s-2$, and $G$ has $r+s+2$ edges. Then $G$ has $2n-x-y-2$ edges, so $G$ is a graph of order $n$ with the maximum possible number of edges among all minimal $(x,y)$ task-dependency graphs of order $n$.

Third, suppose that $p = 1$, $q > 1$, and the initial vertex in $C$ is adjacent to $r \geq 0$ terminal vertices. Then $G$ has $n-q-r-2$ isolated vertices, $x = n-q-r-1$, $y = n-2$, and $G$ has $q+r+1$ edges. Then $G$ has $2n-x-y-2$ edges, so $G$ is a graph of order $n$ with the maximum possible number of edges among all minimal $(x,y)$ task-dependency graphs of order $n$. The proof is analogous for the case when $q = 1$ and $p > 1$.
\end{proof}

In Theorem~\ref{structure-thm}, note that the directed graphs with the structure in case (1) have maximum directed path length $0$ and the directed graphs with the structures in cases (2) and (3)(a) have maximum directed path length $1$. The directed graphs with the structures in cases (3)(b) and (3)(c) have maximum directed path length $2$. Thus we have the following corollary.

\begin{cor}
    Every minimal $(x,y)$ path-dependency graph of order $n$ with the maximum number of edges has no directed path with length more than $2$.
\end{cor}

For minimal task-dependency graphs with only one initial task, we obtain the maximum possible number of edges as a function of the order and the number of terminal tasks as a corollary of Theorem~\ref{mainthm}. We also obtain the corresponding result for task-dependency graphs with only one terminal task.

\begin{cor}
The maximum number of edges of a minimal $(x,1)$ task-dependency graph of order $n$ is $0$ if $n = x$ (in which case $x = 1$), $x$ if $n=x+1$, and $2n-x-3$ if $n>x+1$. The maximum number of edges of a minimal $(1,y)$ task-dependency graph of order $n$ is $0$ if $n = y$ (in which case $y = 1$), $y$ if $n=y+1$, and $2n-y-3$ if $n>y+1$.   
\end{cor}

The next corollary covers the task-dependency graphs that were investigated in \cite{randgbsd}.

\begin{cor}\label{11_edge}
The maximum number of edges of a minimal $(1,1)$ task-dependency graph of order $n$ is $0$ if $n = 1$, $1$ if $n=2$, and $2n-4$ if $n>2$.
\end{cor}


Define a \textit{connected minimal $(x,y)$ task-dependency graph} to be a minimal $(x,y)$ task-dependency graph whose underlying undirected graph is connected. 

\begin{thm}\label{connthm}
For $x,y\geq1$, the maximum number of edges of a connected minimal $(x,y)$ task-dependency graph of order $n$ is $\max(x, y)$ if $n = x+y$ (in which case $\min(x, y) = 1$) and $2n-x-y-2$ if $n>x+y$.
\end{thm}

\begin{proof}
    Without loss of generality, suppose that $x \ge y$. First, note that there are no connected minimal $(x,y)$ task-dependency graphs of order $n$ if $n < x+y$. This is because if $n < x+y$, then some initial vertex must also be a terminal vertex, so there must be an isolated vertex and the resulting directed graph is not connected. 
    
    Next, suppose that $n = x+y$. If $y > 1$, then the maximum possible number of edges in a (not necessarily connected) minimal $(x,y)$ task-dependency graph of order $n$ is $2n-x-y-2 = x+y-2$ by Theorem~\ref{mainthm}. This directed graph cannot be connected since it does not have enough edges, so there is no connected minimal $(x,y)$ task-dependency graphs of order $n$ if $n = x+y$ and $y > 1$. If $y = 1$, then the maximum possible number of edges in a (not necessarily connected) minimal $(x,y)$ task-dependency graph of order $n$ is $2n-x-y-1 = x+y-1$ by Theorem~\ref{mainthm}. Note that this upper bound is attained by the connected minimal $(x,y)$ task-dependency graphs of order $n$ with $x$ initial vertices, one terminal vertex, and edges from every initial vertex to the terminal vertex.

    Finally, suppose that $n > x+y$. By Thereom~\ref{mainthm}, the maximum possible number of edges in a (not necessarily connected) minimal $(x,y)$ task-dependency graph of order $n$ is $2n-x-y-2$. Note that this upper bound of $2n-x-y-2$ edges is attained by the directed graphs $T_{x, y, n}$ for $n \geq x+y+1$.
\end{proof}

\section{Edge-addition process}\label{eap}

The edge-addition process clearly results in an $(x, y)$ task-dependency graph when $x = y = 1$, since each edge addition will only decrease or maintain the numbers of initial vertices and terminal vertices until both are $1$. However, the process does not necessarily result in an $(x, y)$ task-dependency graph whenever $y \neq x$. 

\begin{prop}
 For $x \neq y$ and $n > \max(x,y)$, the $(x, y)$ edge-addition process does not necessarily result in an $(x, y)$ task-dependency graph.    
\end{prop}

\begin{proof}
    Without loss of generality, suppose that $n > x > y$. Let $G_{x, y, n}$ be the task-dependency graph on $\left\{1,2,\dots, n\right\}$ where vertices $n-y+1, \dots, n$ are isolated, $\left\{x-y+1,\dots, n-y\right\}$ are maximally connected, and there are edges $(a, b)$ for all $a \in \left\{1, \dots, x-y\right\}$ and $b \in \left\{x-y+1,\dots, n-y\right\}$. Then $G_{x,y,n}$ has $x$ initial vertices and $y+1$ terminal vertices, since the $y$ vertices $n-y+1, \dots, n$ are both initial and terminal. If we add any edge, it would result in a task-dependency graph with only $x-1$ initial vertices. Thus it is impossible for the $(x, y)$ edge-addition process to result in an $(x, y)$ task-dependency graph if any round produces $G_{x,y,n}$. 
\end{proof}

Given the last result, it is natural to ask whether the $(x, x)$ edge-addition process does not necessarily result in an $(x, x)$ task-dependency graph for $x > 1$. We show in the next proposition that as in the case when $x = y = 1$, the $(x, x)$ edge-addition process always results in an $(x, x)$ task-dependency graph for all $x > 1$. Note how this is the same as with the $(x, x)$ edge-removal process, which always results in an $(x, x)$ task-dependency graph for all $x \ge 1$.

\begin{prop}\label{xxedgeadd}
    For all $x \geq 1$, the $(x, x)$ edge-addition process always results in an $(x, x)$ task-dependency graph.
\end{prop}

\begin{proof}
    Suppose that we completely run the $(x, x)$ edge-addition process and obtain the task-dependency graph $G$ with $r$ initial vertices and $s$ terminal vertices. To show that $r = s = x$, assume for contradiction that $r > x$ or $s > x$.  Without loss of generality, let $r > x$. Note that we must have $s = x$, or else we could add another edge and the resulting task-dependency graph would still have at least $x$ initial vertices and at least $x$ terminal vertices. Analogously to Proposition~\ref{xxedgeremove}, we first argue that $G$ must have at least $2$ non-isolated initial vertices in the same component. Indeed, if every component of $G$ had at most one initial vertex, then we would have $r \leq s$ since every component must also have at least one terminal vertex, contradicting the assumption that $r > x$ and $s = x$. Thus some component of $G$ has multiple non-isolated initial vertices. 

Label these non-isolated initial vertices $u, v$ with $u < v$. If we add the edge $(u, v)$ to $G$, then we obtain a task-dependency graph $G’$ with $r-1$ initial vertices and $s$ terminal vertices. Since $r > x$ and $s = x$, $G’$ has at least $x$ initial vertices and $x$ terminal vertices. This contradicts our assumption that the $(x, x)$ edge-addition process terminated on the task-dependency graph $G$, since it was still possible to add an edge while keeping the number of initial vertices at least $x$ and the number of terminal vertices at least $x$. Thus the $(x, x)$ edge-addition process must result in an $(x, x)$ task-dependency graph.
\end{proof}

When $x = y = 1$, the minimum possible number of edges until the $(x, y)$ edge-addition process on $n$ vertices terminates is $n-1$, since there must be at least $n-1$ edges to include all of the vertices, and a directed path of order $n$ has $n-1$ edges. Note that this is the same as the minimum possible number of edges until the $(x, y)$ edge-removal process on $n$ vertices terminates, and the justification is the same. Using the same proof as Proposition~\ref{eamin} except with \textit{edge-addition} replacing \textit{edge-removal}, we obtain the following generalization. 

\begin{prop}\label{minedgeadd}
    For all $n > \max(x, y)$, the minimum possible number of edges in any $(x,y)$ task-dependency graph produced by the $(x,y)$ edge-addition process on $n$ vertices is $n-\min(x, y)$.
\end{prop}

It is not difficult to use the last proposition to show that the minimum possible number of edges in any task-dependency graph produced by the $(x,y)$ edge-addition process is $n-\min(x, y)$ for $n > \min(x, y)+1$. In the next result, note that we are no longer requiring the task-dependency graph to be an $(x, y)$ task-dependency graph. 

\begin{prop}
    For all $n > \max(x, y)+1$, the minimum possible number of edges in any task-dependency graph produced by the $(x,y)$ edge-addition process on $n$ vertices is $n-\min(x, y)$.
\end{prop}

\begin{proof}
    Let $G$ be a task-dependency graph produced by the $(x,y)$ edge-addition process. Then $G$ is an $(r, s)$ task-dependency graph for some $r \ge x$ and $s \ge y$. If $r = x$ and $s = y$, then $G$ has at least $n-y$ edges by Proposition~\ref{minedgeadd}. Otherwise suppose that $r \ge x+1$ or $s \ge y+1$. If $G$ had any interior vertex $u$, then we could add an edge between two vertices adjacent to $u$ and maintain the number of initial vertices and terminal vertices, so $G$ would not be a task-dependency graph produced by the $(x,y)$ edge-addition process. Thus $G$ has no interior vertices, so all vertices in $G$ are isolated vertices, non-isolated initial vertices, or non-isolated terminal vertices. 
    
    First consider the case when $r \ge x+1$. In this case, we must have $s = y$. $G$ cannot have more than one non-isolated initial vertex, or else we could put an edge between two non-isolated initial vertices and decrease the number of initial vertices by $1$ without changing the number of terminal vertices. Thus the total number of vertices in $G$ is at most $y+1$, contradicting our assumption that $n > \max(x,y)+1$. Now consider the case when $s \ge y+1$. In this case, we must have $r = x$. Analogously to the last case, $G$ cannot have more than one non-isolated terminal vertex, or else we could put an edge between two non-isolated terminal vertices and decrease the number of terminal vertices by $1$ without changing the number of initial vertices. Thus the total number of vertices in $G$ is at most $x+1$, contradicting our assumption that $n > \max(x,y)+1$.
\end{proof}

The maximum possible number of edges until the $(1,1)$ edge-addition process on $n$ vertices terminates is $\binom{n}{2}$, which is much greater than the maximum possible number of edges until the $(1,1)$ edge-removal process on $n$ vertices terminates. Clearly the upper bound follows since there are only $\binom{n}{2}$ edges of the form $(a, b)$ with $a < b$ in a directed graph with vertex set $\left\{1,2,\dots, n\right\}$. To see the lower bound of $\binom{n}{2}$, consider the task-dependency graph $G$ obtained from a complete task-dependency graph with vertices $1,2,\dots, n$ by removing the edge $(n-1,n)$. $G$ has two terminal vertices $n-1$ and $n$, so any instance of the random process that produces the graph $G$ would terminate with $\binom{n}{2}$ edges. In the next few proofs, we generalize this to all $x, y \ge 1$ by showing that the maximum possible number of edges in any $(x,y)$ task-dependency graph produced by the $(x,y)$ edge-addition process is $\binom{n}{2}+1-\binom{\max(x,y)}{2}-\binom{\min(x,y)+1}{2}$ for all $n > x+y$. We start with a definition and a lemma. We define $Q_{x,y,n}$ to be the task-dependency graph on $n$ vertices $1, 2, \dots, n$ with $k = \max{(0,x+y-n)}$ isolated vertices $1, \dots, k$ and all possible edges of the form $(a, b)$ with $k+1 \le a < b \le n$ except when $k+1 \le a < b \le x$ or $n-y+k+1 \le a < b \le n$.

\begin{lem}\label{maxedgelemma}
For all $x,y\geq 1$ and $n\geq \max(x,y)$, the maximum possible number of edges in an $(x,y)$ task-dependency graph of order $n$ is $\binom{n-\max{(0,x+y-n)}}{2}-\binom{x-\max{(0,x+y-n)}}{2}-\binom{y-\max{(0,x+y-n)}}{2}$.
\end{lem}
\begin{proof}
The lower bound follows by using $Q_{x,y,n}$. For $k = \max{(0,x+y-n)}$, the $x-k$ non-isolated initial vertices in $Q_{x,y,n}$ are $k+1, \dots, x$ and the $y$ non-isolated terminal vertices are $n-y+k+1, \dots, n$. $Q_{x,y,n}$ has $\binom{n-k}{2}-\binom{x-k}{2}-\binom{y-k}{2}$ edges. For the upper bound, note that in any task-dependency graph there can be no edges between the initial vertices and no edges between the terminal vertices. Suppose that an $(x,y)$ task-dependency graph $G$ of order $n$ has $z$ isolated vertices, then it has $x-z$ non-isolated initial vertices and $y-z$ non-isolated terminal vertices. Note that $z\geq x+y-n$, since the number of interior vertices in $G$ is $n-(x+y-z)$. The maximum possible number of edges in $G$ is $f(n,x,y,z)=\binom{n-z}{2}-\binom{x-z}{2}-\binom{y-z}{2}$. Since
$$
f(n,x,y,z)-f(n,x,y,z+1)=n-z-1-(x-z-1)-(y-z-1)=z-(x+y-n)+1>0,
$$
the number of edges is maximized when $z$ is minimized.
\end{proof}

\begin{prop}\label{eamax}
    For all $x\ge y \ge 1$ and $n > x+y$, the maximum possible number of edges in any $(x,y)$ task-dependency graph of order $n$ produced by the $(x,y)$ edge-addition process is $\binom{n}{2}+1-\binom{x}{2}-\binom{y+1}{2}$.
\end{prop}
\begin{proof}
    For the lower bound, consider $Q_{x,y+1,n}$. The $x$ initial vertices in $Q_{x,y+1,n}$ are $1, \dots, x$ and the $y+1$ terminal vertices are $n-y, \dots, n$. $Q_{x,y+1,n}$ has $\binom{n}{2}-\binom{x}{2}-\binom{y+1}{2}$ edges. If we randomly add a new edge $(a, b)$ to $G$, then either $1 \le a < b \le x$ or $n-y \le a < b \le n$. In the former case, the edge addition is cancelled since it would result in a task-dependency graph with fewer than $x$ initial vertices. In the latter case, the resulting task-dependency graph has $x$ initial vertices and $y$ terminal vertices, which would cause the $(x, y)$ edge-addition process to halt. Thus the maximum possible number of edges until the $(x, y)$ edge-addition process on $n$ vertices terminates is at least $\binom{n}{2}+1-\binom{y+1}{2}-\binom{x}{2}$.

    For the upper bound, suppose that we are running the $(x,y)$ edge-addition process and it has just terminated with the task-dependency graph $H$. Let $H'$ be the task-dependency graph in the round before the process terminated. Then $H'$ either has (a) $x$ initial vertices and $y+1$ terminal vertices, (b) $x+1$ initial vertices and $y$ terminal vertices, or (c) $x+1$ initial vertices and $y+1$ terminal vertices. If $I$ is the set of initial vertices and $T$ is the set of terminal vertices in $H'$, then the number of edges in $H'$ is at most $\binom{n}{2}-\binom{|I|}{2}-\binom{|J|}{2}$ by Lemma~\ref{maxedgelemma}. Among the cases (a), (b), and (c), this quantity is maximized in case (a) when $|I| = x$ and $|J| = y+1$, which follows from the fact that $\binom{x+1}{2}-\binom{x}{2} \ge \binom{y+1}{2}-\binom{y}{2}$. Thus $H$ has at most $\binom{n}{2}+1-\binom{y+1}{2}-\binom{x}{2}$ edges.
\end{proof}

From the last result, we can obtain the same value for the maximum possible number of edges in any task-dependency graph produced by the $(x,y)$ edge-addition process, without the restriction that the task-dependency graph must be an $(x, y)$ task-dependency graph.

\begin{thm}
For all $x\geq y \ge 1$ and $n > x+y$, the maximum possible number of edges in any task-dependency graph produced by the $(x,y)$ edge-addition process is $\binom{n}{2}+1-\binom{x}{2}-\binom{y+1}{2}$.
\end{thm}

\begin{proof}
    Suppose that $G$ is a task-dependency graph produced by the $(x,y)$ edge-addition process on $n$ vertices. Since the edge-addition process never increases the number of initial vertices or terminal vertices in any round, $G$ must be an $(r,s)$ task-dependency graph for some $r \ge x$ and $s \ge y$. If $r = x$ and $s = y$, then $G$ has at most $\binom{n}{2}+1-\binom{x}{2}-\binom{y+1}{2}$ edges by Proposition~\ref{eamax}. Otherwise $r \ge x+1$ or $s \ge y+1$. Then $G$ has at most $\binom{n}{2}-\binom{x}{2}-\binom{y+1}{2}$ edges or at most $\binom{n}{2}-\binom{x+1}{2}-\binom{y}{2}$ edges by Lemma~\ref{maxedgelemma}. Note that \[\binom{n}{2}-\binom{x}{2}-\binom{y+1}{2} \ge \binom{n}{2}-\binom{x+1}{2}-\binom{y}{2}\] since $\binom{x+1}{2}-\binom{x}{2} \ge \binom{y+1}{2}-\binom{y}{2}$. Thus if $G$ is not an $(x,y)$ task-dependency graph, the number of edges in $G$ is strictly less than $\binom{n}{2}+1-\binom{x}{2}-\binom{y+1}{2}$.
\end{proof}

For the $(x, y)$ edge-addition process on $n$ vertices, it is not difficult to see that the expected number of edges in the resulting task-dependency graph is $\Theta(n^2)$, as we prove in the following proposition. Note that this is much greater than the $\Theta(n)$ expected number of edges in the task-dependency graph obtained from the $(x, y)$ edge-removal process.

\begin{prop}\label{exp_edge_add}
For fixed positive integers $x$ and $y$, the expected number of edges in the $(x, y)$ edge-addition process on $n$ vertices is $\Theta(n^2)$, where the constant in the lower bound depends on $(x,y)$.
\end{prop}

\begin{proof}
The upper bound of $O(n^2)$ is immediate since there are only $\binom{n}{2}$ edges of the form $(a, b)$ with $a < b$ in a directed ordered graph of order $n$. For the lower bound of $\Omega(n^2)$, let $S$ be the set of edges of the form $(a, b)$ with $a < b$ in a directed graph with vertex set $\left\{1,2,\dots, n\right\}$, let $P_x$ be the set of edges of the form $(r, s)$ where $1\leq r<s\leq x+1$, let $P_y$ be the set of edges of the form $(p, q)$ where $n-y\leq p<q\leq n$, and let $P = P_x \cup P_y$. We first note that the process cannot terminate until some edge in $P$ is included, since otherwise all vertices in $[1, x+1]$ would be initial vertices and all vertices in $[n-y,n]$ would be terminal vertices. If that was true, then we would have $x+1$ initial vertices and $y+1$ terminal vertices, so the $(x, y)$ edge-addition process would not terminate yet since any edge addition can only decrease the number of initial vertices and terminal vertices by at most one each. To get a lower bound on the expected number of edges when the process terminates, we can calculate the expected number of rounds until some edge in $P$ is added to the graph. 

Observe that in our random edge-addition process, any permutation of the elements of $S$ is equally likely until the first round in which an element of $P$ is added to the task-dependency graph. In order to bound the expected number of rounds until the first element of $P$ is added to the task-dependency graph, we consider the modified process in which we generate any permutation of the elements of $S$ with equal likelihood, and then we bound the expected location of the first element of $P$ in the permutation. Note that the expected location of the first element of $P$ in the random permutation generated by the modified process is equal to the expected round in which the first element of $P$ is added to the task-dependency graph in the $(x, y)$ edge-addition process, since any permutation of the elements of $S$ is equally likely until the first round in which an element of $P$ is added to the task-dependency graph.

Given any permutation of $S$, the edges in $P$ partition the elements of $S-P$ into $|P|+1$ sequences of $(a_0,a_1,\ldots,a_{|P|})$ consecutive edges with \[\sum_{i = 0}^{|P|} a_i = |S|-|P|.\] Our goal is to determine the expectation of $a_0$. By symmetry, for any permutation $\pi$ of $|P|+1$ elements, the probability that the edges in $P$ partition the elements of $S-P$ into $|P|+1$ sequences of $(a_0,a_1,\ldots,a_{|P|})$ consecutive edges is equal to the probability that the edges in $P$ partition the elements of $S-P$ into $|P|+1$ sequences of $\pi(a_0,a_1,\ldots,a_{|P|})$ consecutive edges. Thus the expectation of $a_0$ is
\[\frac{|S|-|P|}{|P|+1}=\Omega(n^2).\]
\end{proof}

The upper bound in the last proof has a leading coefficient of $\frac{1}{2}$ for all $x, y \ge 1$. On the other hand, the lower bound has a leading coefficient of \[\frac{1}{(x+1)x+(y+1)y+2}.\] For example when $x = y = 1$, there is a multiplicative gap of $3$ between the leading coefficient of $\frac{1}{2}$ in the upper bound and the leading coefficient of $\frac{1}{6}$ in the lower bound.

As a result of the handshake lemma for directed graphs, we obtain the following corollary on the expected maximum in-degree and expected maximum out-degree for the edge-addition process on $n$ vertices.

\begin{cor}
    For fixed positive integers $x$ and $y$, the expected maximum in-degree of the task-dependency graph generated by the $(x, y)$ edge-addition process on $n$ vertices is $\Theta(n)$. Moreover, the expected maximum out-degree of the task-dependency graph generated by the $(x, y)$ edge-addition process on $n$ vertices is $\Theta(n)$.
\end{cor}

\begin{proof}
    The upper bound of $O(n)$ for both the maximum in-degree and maximum out-degree follows since the task-dependency graph has order $n$. For the lower bound, note that the expected average in-degree and expected average out-degree are both $\frac{m}{n}$, where $m$ is the expected number of edges. By Proposition~\ref{exp_edge_add}, the expected average in-degree and expected average out-degree are both $\Theta(n)$, so the expected maximum in-degree and expected maximum out-degree are both $\Theta(n)$.
\end{proof}

Next we show that the expected number of isolated vertices in the task-dependency graph generated by the $(x, y)$ edge-addition process on $n$ vertices is $o(1)$. In other words, the expected number of isolated vertices approaches $0$ as $n$ goes to infinity.

\begin{thm}\label{edge_add_isolated}
    For fixed positive integers $x$ and $y$, the expected number of isolated vertices in the task-dependency graph generated by the $(x, y)$ edge-addition process on $n$ vertices is $o(1)$.
\end{thm}

\begin{proof}
    Suppose that $n$ is sufficiently large so that $\binom{n}{2}-(n-1)-\lceil n \sqrt{n} \rceil > \binom{x+1}{2}+\binom{y+1}{2}$ and $n > \binom{x+1}{2}+\binom{y+1}{2}+1$. As in the proof of Proposition~\ref{exp_edge_add}, let $P_x$ be the set of edges of the form $(r, s)$ where $1\leq r<s\leq x+1$, let $P_y$ be the set of edges of the form $(p, q)$ where $n-y\leq p<q\leq n$, and let $P = P_x \cup P_y$. Note that in the $(x, y)$ edge-addition process on $n$ vertices, every remaining edge is equally likely to be added in every round before the first round in which an element of $P$ is added. We first show that the probability of adding an edge in $P$ in the first $\lceil n \sqrt{n} \rceil$ rounds is $o(1)$. Then we use this to show that the expected number of isolated vertices is $o(1)$.

    First, note that $|P| = \binom{x+1}{2}+\binom{y+1}{2} = O(1)$. Thus, the probability that no edge of $P$ is added in the first $\lceil n \sqrt{n} \rceil$ rounds is \[\prod_{i = 1}^{\lceil n \sqrt{n} \rceil} \frac{\binom{n}{2}-|P|+1-i}{\binom{n}{2}+1-i} \ge \left(1-\frac{|P|}{\binom{n}{2}+1-\lceil n \sqrt{n} \rceil} \right)^{\lceil n \sqrt{n} \rceil} = 1-o(1).\] Therefore, the probability that an edge of $P$ is added in the first $\lceil n \sqrt{n} \rceil$ rounds is $o(1)$.

    The expected number of isolated vertices is $p_1 E_1 + p_2 E_2$ where $p_1$ is the probability of adding an edge in $P$ in the first $\lceil n \sqrt{n} \rceil$ rounds, $E_1$ is the expected number of isolated vertices conditioned on the fact that an edge in $P$ was added in the first $\lceil n \sqrt{n} \rceil$ rounds, $p_2$ is the probability of not adding an edge in $P$ in the first $\lceil n \sqrt{n} \rceil$ rounds (so $p_2 = 1-p_1$), and $E_2$ is the expected number of isolated vertices conditioned on the fact that an edge in $P$ was not added in the first $\lceil n \sqrt{n} \rceil$ rounds. We showed that $p_1 = o(1)$ and $p_2 = 1-o(1)$, and $E_1 \le \max(x, y) = O(1)$, so $p_1 E_1 = o(1)$. 

    Suppose that no edge in $P$ was added in the first $\lceil n \sqrt{n} \rceil$ rounds, so all remaining edges besides the edges in $P$ are equally likely to be added in each of the first $\lceil n \sqrt{n} \rceil$ rounds. Below we bound $E_2$ from above by the expected number of isolated vertices after the first $\lceil n \sqrt{n} \rceil$ rounds. For any vertex $v$ in the task-dependency graph, denote the set of edges in $P$ incident with $v$ by $Q_v$. Then the probability that no edge incident with $v$ is added in the first $\lceil n \sqrt{n} \rceil$ rounds is at most \[\prod_{i = 1}^{\lceil n \sqrt{n} \rceil} \frac{\binom{n}{2}-(n-1)-|P|+|Q_v|+1-i}{\binom{n}{2}-|P|+1-i} \le \left(1-\frac{n-1-|Q_v|}{\binom{n}{2}-|P|}\right)^{\lceil n \sqrt{n} \rceil} \le e^{-\frac{n \sqrt{n}(n-1-|P|)}{\binom{n}{2}-|P|}}.\] Thus by linearity of expectation, \[E_2 \le n e^{-\frac{n \sqrt{n}(n-1-|P|)}{\binom{n}{2}-|P|}} = o(1).\] Therefore, $p_1 E_1 + p_2 E_2 = o(1)$.
\end{proof}


\section{Experimental Results}\label{s:experiment}

We implemented the $(x, y)$ edge-removal and $(x, y)$ edge-addition processes in Python \cite{gt23code}. After running thousands of trials for each process, we have several conjectures about the expected number of edges, expected maximum directed path length, and probability of generating an $(x, y)$ task-dependency graph for both processes. Table~\ref{xy_remove_prob} shows the results of running $1000$ trials of the $(x, y)$ edge-removal process for $(x, y) = (1, 2), (1, 3), (1, 4), (2, 3), (2, 4), (3, 4)$ and finding the ratio of the number of trials on which the process generates an $(x, y)$ task-dependency graph over the total number of trials. We let the number of vertices range from $5$ to $14$. Recall that the $(x, x)$ edge-removal process generates an $(x, x)$ task-dependency graph for each $x \ge 1$, but it is possible for the $(x, y)$ edge-removal process to generate a task-dependency graph with less than $x$ initial vertices or less than $y$ terminal vertices whenever $x \neq y$.

\begin{table}
\begin{tabular}{|c|c|c|c|c|c|c|c|c|c|c|} 
  \hline
  $(x, y)$ & $5$ & $6$ & $7$ & $8$ & $9$ & $10$ & $11$ & $12$ & $13$ & $14$\\ 
  \hline

(1, 2) & 0.947 & 0.993 & 0.999 & 1.000 & 1.000 & 1.000 & 1.000 & 1.000 & 1.000 & 1.000 \\ 
\hline
(1, 3) & 0.507 & 0.748 & 0.908 & 0.965 & 0.992 & 0.998 & 1.000 & 1.000 & 1.000 & 1.000 \\ 
\hline
(1, 4) & 0.051 & 0.229 & 0.484 & 0.733 & 0.883 & 0.971 & 0.985 & 0.998 & 0.998 & 1.000 \\ 
\hline
(2, 3) & 0.687 & 0.870 & 0.968 & 0.994 & 0.998 & 1.000 & 1.000 & 1.000 & 1.000 & 1.000 \\ 
\hline
(2, 4) & 0.086 & 0.332 & 0.572 & 0.806 & 0.926 & 0.978 & 0.990 & 0.996 & 1.000 & 1.000 \\ 
\hline
(3, 4) & 0.258 & 0.590 & 0.796 & 0.908 & 0.981 & 0.992 & 0.999 & 0.999 & 1.000 & 1.000 \\ 
\hline

\end{tabular}
\caption{\label{xy_remove_prob} This table shows the results of running $1000$ trials of the $(x, y)$ edge-removal process on $n$ vertices for $(x, y) = (1, 2), (1, 3), (1, 4), (2, 3), (2, 4), (3, 4)$ and finding the ratio of the number of trials on which the process generates an $(x, y)$ task-dependency graph over the total number of trials. The number of vertices $n$ ranges from $5$ to $14$.}
\end{table}

Based on the results in Table~\ref{xy_remove_prob}, we conjecture for each fixed pair $x, y \ge 1$ that the probability of obtaining an $(x, y)$ task-dependency graph at the end of the $(x, y)$ edge-removal process approaches $1$ as the number of vertices goes to infinity. 

\begin{conj}
For fixed $x, y \ge 1$, if $r_{x,y,n}$ denotes the probability that the $(x, y)$ edge-removal process on $n$ vertices generates an $(x, y)$ task-dependency graph, then \[\lim_{n \rightarrow \infty} r_{x, y, n} = 1.\]
\end{conj}

Analogously to the previous experiment, Table~\ref{xy_add_prob} shows the results of running $1000$ trials of the $(x, y)$ edge-addition process for $(x, y) = (1, 2), (1, 3), (1, 4), (2, 3), (2, 4), (3, 4)$ and finding the ratio of the number of trials on which the process generates an $(x, y)$ task-dependency graph over the total number of trials. As with the last table, we let the number of vertices range from $5$ to $14$. Recall that the $(x, x)$ edge-addition process generates an $(x, x)$ task-dependency graph for each $x \ge 1$, but it is possible for the $(x, y)$ edge-addition process to generate a task-dependency graph with more than $x$ initial vertices or more than $y$ terminal vertices whenever $x \neq y$.

\begin{table}
\begin{tabular}{|c|c|c|c|c|c|c|c|c|c|c|} 
  \hline
  $(x, y)$ & $5$ & $6$ & $7$ & $8$ & $9$ & $10$ & $11$ & $12$ & $13$ & $14$\\ 
  \hline

(1, 2) & 0.923 & 0.962 & 0.964 & 0.98 & 0.993 & 0.989 & 0.998 & 0.995 & 1.000 & 1.000 \\ 
\hline
(1, 3) & 0.715 & 0.828 & 0.903 & 0.914 & 0.954 & 0.968 & 0.978 & 0.988 & 0.982 & 0.992 \\ 
\hline
(1, 4) & 0.382 & 0.616 & 0.727 & 0.825 & 0.864 & 0.916 & 0.931 & 0.937 & 0.958 & 0.963 \\ 
\hline
(2, 3) & 0.958 & 0.986 & 0.988 & 0.998 & 0.998 & 0.999 & 1.000 & 1.000 & 1.000 & 1.000 \\ 
\hline
(2, 4) & 0.706 & 0.890 & 0.954 & 0.981 & 0.985 & 0.988 & 0.994 & 0.994 & 0.999 & 0.999 \\ 
\hline
(3, 4) & 0.907 & 0.982 & 0.994 & 0.999 & 1.000 & 0.999 & 1.000 & 1.000 & 1.000 & 1.000 \\ 
\hline

\end{tabular}
\caption{\label{xy_add_prob} This table shows the results of running $1000$ trials of the $(x, y)$ edge-addition process on $n$ vertices for $(x, y) = (1, 2), (1, 3), (1, 4), (2, 3), (2, 4), (3, 4)$ and finding the ratio of the number of trials on which the process generates an $(x, y)$ task-dependency graph over the total number of trials. The number of vertices $n$ ranges from $5$ to $14$.}
\end{table}

As with the previous experiment on the $(x, y)$ edge-removal process, we conjecture for each fixed pair $x, y \ge 1$ that the probability of obtaining an $(x, y)$ task-dependency graph at the end of the $(x, y)$ edge-addition process approaches $1$ as the number of vertices goes to infinity. 

\begin{conj}\label{prob_add}
For fixed $x, y \ge 1$, if $a_{x,y,n}$ denotes the probability that the $(x, y)$ edge-addition process on $n$ vertices generates an $(x, y)$ task-dependency graph, then \[\lim_{n \rightarrow \infty} a_{x, y, n} = 1.\]
\end{conj}

We were able to prove Conjecture~\ref{prob_add} using an approach similar to the proofs of Propositions~\ref{xxedgeremove} and \ref{xxedgeadd}. 

\begin{thm}
Conjecture~\ref{prob_add} is true.
\end{thm}

\begin{proof}
Suppose that $G$ is an $(r, s)$ task-dependency graph generated by the $(x, y)$ edge-addition process on $n$ vertices with $(r, s) \neq (x, y)$. Then we must have $r > x$ or $s > y$, so without loss of generality, let $r > x$. Note that we must have $s = y$, or else we could add another edge and the resulting task-dependency graph would still have at least $x$ initial vertices and at least $y$ terminal vertices. We first claim that $G$ cannot have any pair of distinct non-terminal vertices $u < v$ without an edge between them. If there were such a pair, then we could add the edge $(u, v)$ without decreasing the number of terminal vertices, so the $(x, y)$ edge-addition process would not have terminated on $G$. 

Next, we claim that $G$ must have some isolated vertex. To see this, suppose that $G$ had no isolated vertices. Then $G$ only has a single initial vertex since every pair of non-terminal vertices must have an edge between them. Thus $r = 1$, which contradicts $r > x$, so $G$ must have some isolated vertex.

By Markov's inequality, the probability that $G$ has an isolated vertex is at most the expected number of isolated vertices, which is $o(1)$ by Theorem~\ref{edge_add_isolated}. Thus, the probability that $G$ is not an $(x, y)$ task-dependency graph is $o(1)$, so $a_{x,y,n} = 1-o(1)$.
\end{proof}

We also ran $1000$ trials of the $(1, 1)$ edge-removal and $(1, 1)$ edge-addition processes and calculated the average number of edges and average maximum directed path length. We ran the processes for the number of vertices ranging from $3$ to $40$. Figure~\ref{xyp_plots} shows plots of the averages with fit curves.

\begin{figure}[h]\label{xyp_plots}
\caption{Plot (a) shows the average number of edges with respect to $n$ of the $(1, 1)$ edge-removal process on $n$ vertices as $n$ ranges from $3$ to $40$. The line of best fit is $y = 1.349x - 2.029$. \\
Plot (b) shows the average maximum directed path length with respect to $n$ of the $(1, 1)$ edge-removal process on $n$ vertices as $n$ ranges from $3$ to $40$. The logarithmic curve of best fit is $y = 3.679 \ln{(x+1.799)} - 3.776$. \\
Plot (c) shows the average number of edges with respect to $n$ of the $(1, 1)$ edge-addition process on $n$ vertices as $n$ ranges from $3$ to $40$. The quadratic of best fit is $y = 0.369 x^2 - 0.258 x + 0.007$. \\
Plot (d) shows the average maximum directed path length with respect to $n$ of the $(1, 1)$ edge-addition process on $n$ vertices as $n$ ranges from $3$ to $40$. The line of best fit is $y = 0.779x - 0.273$. All coefficients have been rounded to $3$ places.}
\includegraphics[width=\textwidth]{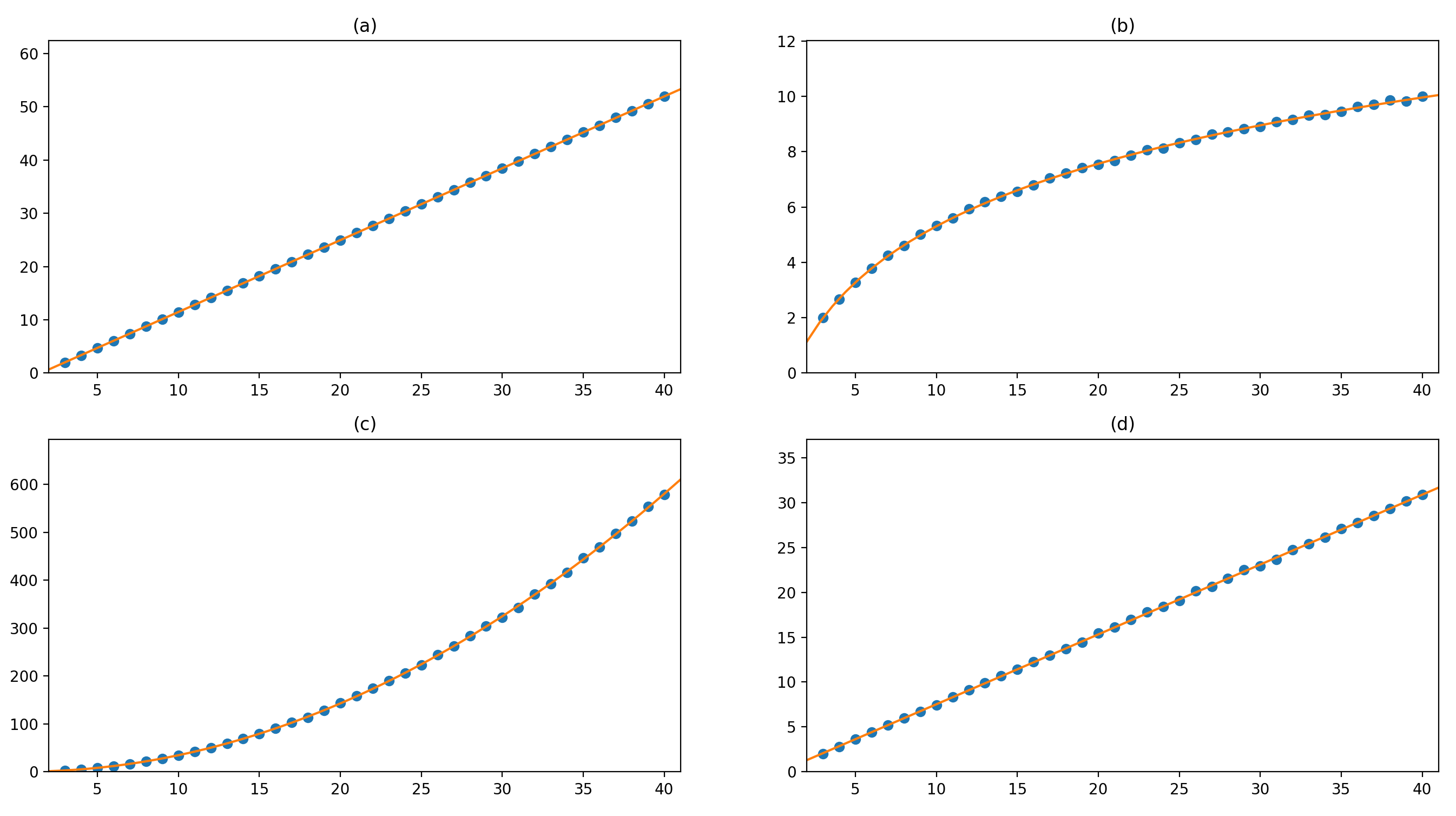}
\end{figure}

From the results in Section~\ref{ex_bounds}, we know that the expected number of edges in the task-dependency graph generated by the $(1, 1)$ edge-removal process on $n$ vertices is $\Theta(n)$, with an upper bound of $2n-4$ and a lower bound of $n-1$. Based on the results of the experiment, we conjecture that the expectation is closer to the lower bound. 

\begin{conj}
If $e_{n}$ denotes the expected number of edges in the task-dependency graph generated by the $(1, 1)$ edge-removal process on $n$ vertices, then \[\lim_{n \rightarrow \infty} \frac{e_n}{n} \approx 1.35.\]    
\end{conj}

\begin{lem}\label{lem:edge-prob}
For any $1\le r<s\le n$, the probability that edge $(r,s)$ remains in the $(1,1)$ task dependency graph generated by $(1,1)$ edge-removal process is at most
$$
\frac{1}{s-1}+\frac{1}{n-r}-\frac{1}{n-2+s-r}.
$$
\end{lem}
\begin{proof}
Recall that starting with a maximally connected $(1,1)$ task dependency graph $G$ of order $n$, the edge-removal process randomly permutes all $\binom{n}{2}$ edges and then removes them one by one except if the removal makes any interior vertex exterior. Consider the edge sets $A_s=\{(p,s):p<s\}$ and $B_r=\{(r,t):r<t\}$. A necessary condition for edge $(r,s)$ to remain in the final task-dependency graph is that in the permutation of edges it has to come last among $B_r$ or $A_s$. Otherwise when edge $(r,s)$ is being considered for removal, vertex $r$ would have out-degree more than $1$, vertex $s$ would have in-degree more than $1$, and $(r,s)$ would have been removed. By the principle of inclusion and exclusion, this necessary condition happens with probability
$$
\frac{1}{|A_s|}+\frac{1}{|B_r|}-\frac{1}{|A_s\cup B_r|}=\frac{1}{s-1}+\frac{1}{n-r}-\frac{1}{n-2+s-r}.
$$
\end{proof}

In the last proof, note that the condition that edge $(r, s)$ comes last among the edges of $B_r$ or $A_s$ is only a necessary condition for $(r, s)$ to remain in the final task-dependency graph. Specifically, this condition is not sufficient for $(r, s)$ to remain in the final task-dependency graph. For example, suppose that there exist edges $(r, t) \in B_r$ and $(p, s) \in A_s$ which do not come last among the edges of $B_r$ and $A_s$ respectively. It is possible for both edges to remain in the final task-dependency graph if $(r, t)$ comes last among the edges of $A_t$ and $(p, s)$ comes last among the edges of $B_p$. In this case, the edge $(r, s)$ would be removed from the task-dependency graph, even if it comes last among the edges of $B_r$ or $A_s$, since its removal would not make any vertex exterior.

\begin{prop}
If $e_{n}$ denotes the expected number of edges in the task-dependency graph generated by the $(1, 1)$ edge-removal process on $n$ vertices, then 
\[\lim_{n \rightarrow \infty} \frac{e_n}{n} \le 3-2\ln 2\approx 1.614.\]
\end{prop}
\begin{proof}
We prove it by summing up the probabilities in Lemma~\ref{lem:edge-prob}:
$$
\frac{e_n}{n}\leq \frac{1}{n}\left(\sum_{r=1}^{n-1}\sum_{s=r+1}^n\frac{1}{s-1}+\frac{1}{n-r}-\frac{1}{n-2+s-r}\right)
=\frac{1}{n}\left(n-1+n-1-\sum_{d=1}^{n-1}\frac{n-d}{n-2+d}\right).
$$
By setting $z=\frac{n-2+d}{n}$, the last term approaches $-\int_{1}^{2}\frac{2-z}{z}dz=1-2\ln 2$. Therefore
$$
\lim_{n\rightarrow\infty}\frac{e_n}{n}\leq 3-2\ln 2.
$$
\end{proof}

From the results in Section~\ref{eap}, we know that the expected number of edges in the task-dependency graph generated by the $(1, 1)$ edge-addition process on $n$ vertices is $\Theta(n^2)$, with an upper bound of approximatey $\frac{1}{2}n^2$ and a lower bound of approximately $\frac{1}{4}n^2$. Based on the results of the experiment, we conjecture that the expectation is slightly closer to the lower bound. 

\begin{conj}
If $f_{n}$ denotes the expected number of edges in the task-dependency graph generated by the $(1, 1)$ edge-addition process on $n$ vertices, then \[\lim_{n \rightarrow \infty} \frac{f_n}{n^2} \approx 0.37.\]    
\end{conj}

The experimental results for the expected maximum directed path length are the most interesting of all. For the $(1, 1)$ edge-addition process, the curve of best fit is clearly linear for the average maximum directed path length with respect to $n$. However for the $(1, 1)$ edge-removal process, it appears that the expected maximum directed path length is logarithmic in $n$.

\begin{conj}
    For the $(1, 1)$ edge-addition process on $n$ vertices, the expected maximum directed path length of the resulting task-dependency graph is $\Theta(n)$.
\end{conj}

\begin{conj}\label{conj_logn}
    For the $(1, 1)$ edge-removal process on $n$ vertices, the expected maximum directed path length of the resulting task-dependency graph is $\Theta(\log{n})$.
\end{conj}

We did not prove Conjecture~\ref{conj_logn}, but we found a related result for a different random process for generating task-dependency graphs. Let the \textit{random directed tree process} on $n$ vertices be the random process which starts with $n$ isolated vertices $1,2,\ldots,n$, and in round $s>0$, an edge $(r,s+1)$ is added uniformly at random where $r\le s$. Note that for any task-dependency graph generated by the random directed tree process on $n$ vertices, the underlying undirected graph is a tree. 

\begin{lem}
Let $G$ be a task dependency graph of order $n$ generated by the random directed tree process on $n$ vertices. Then for $k>1$ the expected length of the unique path from vertex $1$ to vertex $k$ is $1+\frac{1}{2}+\ldots+\frac{1}{k-1}=\Theta(\log k)$.
\end{lem}
\begin{proof}
Denote by $f(k)$ the expected length of the unique path from vertices $1$ to $k$. The probability that vertices $r<s$ are adjacent is $\frac{1}{s-1}$, so we have $f(1)=0$ and for $k>1$
$$
f(k)=1+\frac{1}{k-1}\left(f(1)+\ldots+f(k-1)\right).
$$
It simplifies to
$$
f(k+1)-f(k)=\frac{1}{k}.
$$
\end{proof}

\section{Future directions and open problems}\label{open_prob}
Our experimental results suggest that both the $(x, y)$ edge-addition process and the $(x, y)$ edge-removal process result in an $(x, y)$ task-dependency graph with probability approaching $1$ as the number of vertices increases. We proved that this is true for the $(x, y)$ edge-addition process, but it remains to prove for the $(x, y)$ edge-removal process. 

We determined the maximum possible number of edges in a minimal $(x, y)$ task-dependency graph of order $n$ for all $x, y \ge 1$ and $n \ge \max(x, y)$, and we showed that this is equal to the maximum possible number of edges in any task-dependency graph generated by the $(x, y)$ edge-removal process on $n$ vertices. We also characterized the extremal task-dependency graphs. 

Many related questions remain. For one, how many ways are there to order the tasks in a given task-dependency graph? In particular, what is the maximum possible number of orderings of the tasks in an $(x,y)$ task-dependency graph of order $n$?

We made some progress on this question in the following theorem. Specifically, we solved it for all $n \ge 1$ when $x = y = 1$. We also solved it for all $x, y \ge 1$ when $\max(x, y) \le n \le \max(x,y)+2$.

\begin{prop}
    For $x \geq y \geq 1$, the maximum possible number of orderings of the tasks in an $(x,y)$ task-dependency graph of order $n$ is $n!$ for $n = x$ (in which case $x = y$), $\frac{(x+1)!}{x-y+2}$ for $n = x+1$, $\frac{(x+2)!}{2(x-y+2)}$ for $n = x+2$ and $y > 1$, and $\frac{(x+2)!}{2(x-y+3)}$ for $n = x+2$ and $y = 1$. The maximum possible number of orderings of the tasks in a $(1,1)$ task-dependency graph of order $n$ is $(n-2)!$ for $n \geq 2$. 
\end{prop}

\begin{proof}
    For $n = x$, the upper bound is immediate since there are a total of $n!$ ways to order the vertices of a directed graph of order $n$. This bound is attained by the directed graph of order $n$ with all vertices isolated.

    For $n = x+1$, we noted in the proof of Theorem~\ref{mainthm} that any $(x,y)$ task-dependency graph has $y-1$ isolated vertices, $x-y+1$ non-isolated initial vertices, one non-isolated terminal vertex, and edges from each non-isolated initial vertex to the non-isolated terminal vertex. Thus any ordering of the vertices that has the non-isolated terminal vertex after all of the non-isolated initial vertices will be consistent with the task-dependency graph. The probability that a random ordering of the vertices has the non-isolated terminal vertex after all of the non-isolated initial vertices is $\frac{1}{x-y+2}$ and there are $n! = (x+1)!$ total orderings, so the number of orderings consistent with the task-dependency graph is \[\frac{n!}{x-y+2} = \frac{(x+1)!}{x-y+2}.\]

    For $n = x+2$, we noted in the proof of Theorem~\ref{mainthm} that there are two possibilities for an $(x,y)$ task-dependency graph $G$ of order $n$: $G$ either has (a) $y-1$ isolated vertices, $x-y+1$ non-isolated initial vertices, one non-isolated terminal vertex, and one interior vertex, or (b) $y-2$ isolated vertices, $x-y+2$ non-isolated initial vertices, two non-isolated terminal vertices, and no interior vertex. In either case $G$ has $2n-x-y-2$ edges. Note that case (b) can only happen when $y > 1$.

    For case (a), suppose that $G$ has $t$ non-isolated initial vertices with edges to the interior vertex and $x-y+1-t$ non-isolated initial vertices with edges to the non-isolated terminal vertex, with $1 \leq t \leq x-y+1$. There are $(y-1)!$ ways to order the isolated vertices, $\binom{x+2}{y-1}$ ways to choose the locations of the isolated vertices, $1$ way to choose the location of the non-isolated terminal vertex, $\binom{x-y+2}{t+1}$ ways to choose the locations of the interior vertex together with the $t$ non-isolated initial vertices that have edges to the interior vertex, $t!$ ways to order the $t$ non-isolated initial vertices that have edges to the interior vertex, and $(x-y+1-t)!$ ways to order the non-isolated initial vertices that have edges to the non-isolated terminal vertex. This gives a total of \[(y-1)!  \binom{x+2}{y-1}  \binom{x-y+2}{t+1} t! (x-y+1-t)!\] orderings of the vertices, which simplifies to \[\frac{(x+2)!}{(x-y+3)(t+1)}.\] This is maximized when $t = 1$, which gives \[\frac{(x+2)!}{2(x-y+3)}.\] Thus for $n = x+2$ and $y = 1$, this quantity is the maximum number of orderings.

    For case (b), suppose that $t$ non-isolated initial vertices have edges to the first non-isolated terminal vertex, with $1 \leq t \leq x-y+1$. This forms a first component of size $t+1$. Furthermore, suppose that $x-y+2-t$ non-isolated initial vertices have edges to the second non-isolated terminal vertex. This forms a second component of size $x-y+3-t$. There are $(y-2)!$ ways to order the isolated vertices, $\binom{x+2}{y-2}$ ways to choose the locations of the isolated vertices, $\binom{x-y+4}{t+1}$ ways to choose the locations of the vertices in the first component, $t!$ ways to order the vertices in the first component, and $(x-y+2-t)!$ ways to order the vertices in the second component. This gives a total of \[(y-2)!  \binom{x+2}{y-2}  \binom{x-y+4}{t+1}  t!  (x-y+2-t)!\] orderings of the vertices, which simplifies to \[\frac{(x+2)!}{(t+1)(x-y+3-t)}.\] This is maximized when $t = 1$, which gives \[\frac{(x+2)!}{2(x-y+2)}.\] Thus for $n = x+2$ and $y > 1$, this quantity is the maximum number of orderings.

If $x = y = 1$ and $n \ge 2$, then there are $n-2$ interior vertices, an initial vertex, and a terminal vertex. The initial task must precede all others, the final task must follow all others, and the number of ways to order the interior tasks is at most $(n-2)!$ since there are $(n-2)!$ ways to order the interior tasks if there are no restrictions caused by edges from one interior task to another. This upper bound of $(n-2)!$ is attained by $T_{1,1,n}$.
\end{proof}

For each of the cases covered in the last theorem, it is notable that the maximum possible number of orderings of the tasks in an $(x,y)$ task-dependency graph of order $n$ is attained by a minimal $(x, y)$ task-dependency graph of order $n$ with the maximum possible number of edges. It would be interesting to determine whether this pattern continues to hold for all $x, y \ge 1$ and $n \ge \max(x, y)$.

There are also some interesting probabilistic questions related to our results. We proved for all $x, y \ge 1$ that the expected number of edges in the random task-dependency graph resulting from the $(x, y)$ edge-removal process on $n$ vertices is $\Theta(n)$. Our upper bound has leading coefficient $2$ and our lower bound has leading coefficient $1$, so it would be interesting to close the gap between these bounds. We also showed for all $x, y \ge 1$ that the expected number of edges in the random task-dependency graph resulting from the $(x, y)$ edge-addition process on $n$ vertices is $\Theta(n^2)$. The upper bound has leading coefficient $\frac{1}{2}$ and the lower bound has leading coefficient $\frac{1}{(x+1)x+(y+1)y+2}$. 

Based on our results, we describe a random process for generating $(x, y)$ task-dependency graphs of order $n$ with $m$ edges for any $x, y \ge 1$, using only edge-addition and edge-removal. In the case that $x = y = 1$, this random process provides an alternative to the one in \cite{randgbsd} for randomly generating $(1, 1)$ task-dependency graphs of order $n$ with $m$ edges.

\begin{prop}
    For any $x, y \ge 1$, $n > \max(x, y)+1$, and $2n-x-y-2 \le m \le \binom{n}{2}-(x+y)n$, there is a random process only involving edge-addition and edge-removal which generates with high probability an $(x, y)$ task-dependency graph of order $n$ with $m$ edges.
\end{prop}

\begin{proof}
    Consider the following random process. First, run the $(x, y)$ edge-addition process on $n$ vertices until it terminates. With high probability, it generates an $(x, y)$ task-dependency graph $G$. Let $r$ be the number of edges in $G$. If $r = m$, then $G$ is the final task-dependency graph. If $r < m$, then $r < \binom{n}{2}-(x+y)n$. Since there are at most $(x+y)n$ missing edges whose addition decreases the number of initial or terminal vertices, there must be a missing edge whose addition does not decrease the number of initial or terminal vertices. Thus in this case, we can add edges to $G$ uniformly at random without decreasing the number of initial vertices and terminal vertices until the resulting task-dependency graph has $m$ edges. 
    
    If $r > m$, then perform the $(x, y)$ edge-removal process starting with $G$ until the resulting task-dependency graph has $m$ edges. Note that running the complete $(x, y)$ edge-removal process on $n$ vertices always results in a task-dependency graph with at most $2n-x-y-2$ edges given the assumption that $n > \max(x, y)+1$. Moreover, for each round in which we perform the $(x, y)$ edge-removal process starting with $G$, the resulting task-dependency graph is guaranteed to be an $(x, y)$ task-dependency graph assuming that $G$ is an $(x, y)$ task-dependency graph. Thus, as long as $2n-x-y-2 \le m \le \binom{n}{2}-(x+y)n$, with high probability this random process will generate an $(x, y)$ task-dependency graph of order $n$ with $m$ edges.
\end{proof}

There are other natural questions about expected values that pertain to both the $(x, y)$ edge-addition process and the $(x, y)$ edge-removal process. 

\begin{enumerate}
\item What is the expected maximum directed path length of the resulting task-dependency graph as a function of $x$, $y$, and $n$? 

\item What is the expected number of orderings of the tasks of the resulting task-dependency graph as a function of $x$, $y$, and $n$?

\item What is the expected number of copies of a given directed acyclic graph $H$ in the resulting task-dependency graph as a function of $x$, $y$, and $n$?
\end{enumerate}

Besides expected values, there are other natural probabilistic questions. For example, find the probability that the resulting task-dependency graph of order $n$ has exactly $m$ edges, the probability that it has maximum directed path length $m$, the probability that the number of orderings is $m$, or the probability that the number of copies of a given directed acyclic graph $H$ is $m$.

\section*{Acknowledgement}
Shen-Fu Tsai is supported by the Ministry of Science and Technology of Taiwan under grant MOST 111-2115-M-008-010-MY2. We thank Lance Menthe for helpful comments on the manuscript.

\end{document}